%% file: main.tex
\documentclass[11pt]{article}

\usepackage[margin=1in]{geometry}
\usepackage{fullpage}
\usepackage{tgtermes}
\usepackage[T1]{fontenc}
\usepackage[colorlinks,citecolor=blue,linkcolor=blue,urlcolor=red]{hyperref}
\usepackage{graphicx}
\usepackage{amsfonts,amsmath,amsthm,amssymb,dsfont,mathtools,multirow}
\mathtoolsset{centercolon}
\usepackage{xfrac,nicefrac}
\usepackage{mathdots}
\usepackage{bm,bbm}
\usepackage{url}
\usepackage{caption}
\usepackage{paralist}
\usepackage{enumerate}
\usepackage[normalem]{ulem}
\usepackage{enumitem}
\usepackage{xspace}
\xspaceaddexceptions{]\}}
\usepackage[capitalise]{cleveref}
\crefname{algocf}{Algorithm}{Algorithms}
\Crefname{algocf}{Algorithm}{Algorithms}
\crefname{claim}{Claim}{Claims}
\usepackage{tabu}
\usepackage{framed}
\usepackage{float,wrapfig}
\usepackage{subfigure}
\usepackage{soul}
\usepackage[usenames,dvipsnames]{pstricks}
\usepackage{thmtools, thm-restate}
\usepackage[linesnumbered,boxed,ruled,vlined]{algorithm2e}
\usepackage{algpseudocode}
\usepackage{tikz}
\usetikzlibrary{decorations.pathreplacing}
\usetikzlibrary{calc}
\usepackage{regexpatch}
\usetikzlibrary{positioning}
\usetikzlibrary{arrows.meta}
\usepackage{comment}

\theoremstyle{plain}

\newtheorem{theorem}{Theorem}[section]
\newtheorem{lemma}[theorem]{Lemma}

\newtheorem{prop}[theorem]{Proposition}
\newtheorem{cor}[theorem]{Corollary}
\newtheorem{claim}[theorem]{Claim}
\theoremstyle{definition}
\newtheorem{definition}[theorem]{Definition}

\newtheorem{problem}{Problem}

\def\sub{\subseteq}
\def\eps{\varepsilon}

\newcommand{\MinPlusConv}{\circledast}

\def\cA{\mathcal{A}}
\def\cB{\mathcal{B}}
\def\cC{\mathcal{C}}
\def\cD{\mathcal{D}}

\newcommand{\RR}{\mathbb{R}}

\newcommand{\EQ}{\;=\;}

\newcommand{\Ot}{\widetilde{O}}
\newcommand{\tO}{\widetilde{O}}

\newcommand{\BF}[3]{\textsc{Bellman-Ford}_{#1}(#2, #3)} 

\begin{document}

\title{All-Hops Shortest Paths}
\author{
Virginia Vassilevska Williams\thanks{Massachusetts Institute of Technology. \texttt{virgi@mit.edu}. Supported by NSF Grant CCF-2330048, BSF Grant 2020356 and a Simons Investigator Award. This work was done in part while the author was visiting the Simons Institute for the Theory of Computing.}
  \and
Zoe Xi\thanks{Massachusetts Institute of Technology. \texttt{zoexi@mit.edu}.}
  \and
Yinzhan Xu\thanks{Massachusetts Institute of Technology. \texttt{xyzhan@mit.edu}. Supported by NSF Grant CCF-2330048, a Simons
Investigator Award and HDR TRIPODS Phase II grant 2217058. }
  \and
Uri Zwick\thanks{Tel Aviv University. \texttt{zwick@tau.ac.il}. This work was done in part while the author was visiting the Simons Institute for the Theory of Computing.}
}
\date{}

\maketitle

\begin{abstract}
Let $G=(V,E,w)$ be a weighted directed graph without negative cycles. For two vertices $s,t\in V$, we let $d_{\le h}(s,t)$ be the minimum, according to the weight function~$w$, of a path from~$s$ to~$t$ that uses at most~$h$ edges, or \emph{hops}. We consider algorithms for computing $d_{\le h}(s,t)$ for every $1\le h\le n$, where $n=|V|$, in various settings. We consider the single-pair, single-source and all-pairs versions of the problem. We also consider a distance oracle version of the problem in which we are not required to explicitly compute all distances $d_{\le h}(s,t)$, but rather return each one of these distances upon request. We consider both the case in which the edge weights are arbitrary, and in which they are small integers in the range $\{-M,\ldots,M\}$. For some of our results we obtain matching conditional lower bounds.
\end{abstract}

\input{intro}

\input{prelim}

\input{s-t-alg-new}

\input{single-source}

\input{all-pair}

\input{distoracles}

\input{lower-bound-weighted}

\bibliographystyle{alpha}
\bibliography{ref}

\appendix
\section{Lower Bound from Triangle Detection to Single Pair All-Hop Shortest Paths.}
\input{BMMlb}
\end{document}

%% file: intro.tex
\section{Introduction}
\label{sec:intro}

The shortest path problem is one of the most fundamental and most studied algorithmic graph problem. In its standard version, we are given a weighted directed graph $G=(V,E,w)$, where $w:E\to\RR$ assigns a \emph{weight} to each edge. For given $s,t\in V$, the goal is to find a \emph{shortest path} from~$s$ to~$t$, i.e., a path that minimizes the sum of the weights of the edges on the path. We let $d(s,t)$ denote the weight of such a shortest path, also known as the \emph{distance} from $s$ to $t$. This is the {\em single-pair} variant in the problem. In the {\em single-source} version we need to find shortest paths from a fixed source~$s$ to all other vertices in the graph. In the \emph{all-pairs} version we need to find shortest paths between all pairs of vertices in the graph.

In the standard version of the problem no constraints are placed on the \emph{unweighted length}, i.e., the number of edges, used in the paths considered. We refer to this number also as the number of \emph{hops} used by the path. Sometimes, a compromise is required between the length of a path, according to the weight function~$w$, and its number of edges, or hop count. The weight of an edge may correspond, for example, to the cost of using the edge, while the length of the edge, in our case $1$, may correspond, for example, to the time, or energy, needed to traverse the edge. 

When the input graph is equipped, in addition to the weight function $w:E\to\RR$, also with a general length function $\ell:E\to\RR^+$, and the goal is to minimize the weight of a path from~$s$ to~$t$ given a constraint on its length, the problem is referred to as the \emph{bicriteria} shortest path problem (see, e.g., Raith and Ehrgott~\cite{RaEh09}). Not surprisingly, the problem is NP-hard (see \cite[Problem {[ND30]} on page 214]{GaJo79}). 

When the length of each edge is~$1$, the problem can of course be solved in polynomial time. In fact, given a fixed source vertex $s\in V$, the classical Bellman-Ford algorithm \cite{Bellman58,Ford56} can be used to compute the minimum weight of an $s$ to $t$ path that uses at most $h$ hops,
$d_{\le h}(s,t)$, for every $1\le h\le n$ and every $t\in V$, in $O(mn)$ time, where $m=|E|$ and $n=|V|$. Recently, Kociumaka and Polak \cite{KoAd23} showed that this is conditionally optimal.

Besides $d_{\le h}(s,t)$, prior work (e.g., \cite{ChAn04}) has also considered the related quantity $d_{h}(s,t)$ defined as the minimum weight of a path from~$s$ to~$t$ that uses \emph{exactly} $h$ hops (or $\infty$ if no such $h$-hop path exists). If one can solve this ``exactly $h$'' version of any of the problems (single-pair, single-source, all-pairs) for all values of $h$, then one can use these values to also compute $d_{\le h}(s,t)=\min_{k\leq h} d_k(s,t)$ very quickly, and hence the ``exactly $h$'' version of the problems is at least as hard as the ``at most~$h$'' version. An alternative way to reduce the ``at most $h$'' version of the problem to the ``exactly $h$'' version of the problem is to add self-loops of weight $0$ to every vertex in the graph, so that for any pair of vertices $s$ and $t$, the minimum weight of an $s$ to $t$ path that uses exactly $h$ hops in the modified graph is exactly the minimum weight of an $s$ to $t$ path that uses at most $h$ hops in the original graph. 

There is a reduction in the other direction as well, provided that large negative edge weights are allowed. Suppose that the weights in the original graph $G=(V,E,w)$ are in $[-M,M]$. Then we can create a new graph $G'=(V,E,w')$ so that $w'(e)=w(e)-2Mn$. Any $h$-hop path would have weight $\leq (M-2Mn)h$ and any $h-1$-hop path would have weight $\geq (h-1)(-M-2Mn)= (M-2Mn)h +2Mn+M-2Mh >(M-2Mn)h$. Thus, $d'_h(u,v)=d'_{\leq h}(u,v)=d_{h}(u,v)-2Mnh$ (here, $d'$ denotes distances in $G'$ and $d$ denotes distances in $G$) and one can recover the exactly-$h$ distances of $G$ from the $\leq h$ all-hops distances of $G'$. For directed acyclic graphs (DAGs), one can use standard potential function reweighting to make sure that the reduction DAG $G'$ has nonnegative edge weights. 

In arbitrary (nonacyclic) graphs and graphs with bounded edge weights, however, the ``exactly $h$'' and the ``at most $h$'' versions of the problem could have different complexities. The two types of distances are also qualitatively different, e.g., while the paths achieving $d_{\leq h}(s,t)$ are without loss of generality simple in graphs without negative cycles, those achieving $d_{h}(s,t)$ need not be. This makes it difficult to use standard shortest path techniques to speed-up the computation of $d_h(s,t)$, as for instance, randomly sampling vertices would not help to hit a length $h$-path if the number of distinct vertices on the path is actually small. 

\subsection{Our results}
\label{sec:results}
We consider many variants of the all-hops shortest paths problem: single pair, single source, all pairs, distance oracles, arbitrary and bounded integer weights.
\paragraph{Single-Pair All-Hops Shortest Paths.}
We first consider the Single-Pair All-Hops Shortest Paths problem: Given $G=(V,E,w)$ and $s,t\in V$, compute $d_{\leq h}(s,t)$, for every $h=1,\ldots,n-1$.
Kociumaka and Polak~\cite{KoAd23} have shown that Bellman-Ford's $O(mn)$ running time is conditionally optimal for the problem under the APSP hypothesis of fine-grained complexity (see \cite{RodittyZ11,vsurvey}).
As their hardness reduction produces graphs with large edge weights, the conditional lower bound may not hold for graphs with small integer weights. 
We thus consider the problem for integer edge weights in $\{-M,\ldots,M\}$ and graphs with no negative cycles (as is common in the shortest paths literature).
We focus on the dense graph setting for which the reduction of \cite{KoAd23} provides a conditional $n^{3-o(1)}$ lower bound for the problem.

In the following, we let $\omega<2.371552$ \cite{WXXZ24} denote the smallest constant for which two $n\times n$ matrices can be multiplied using $n^{\omega+o(1)}$ arithmetic operations. For conciseness, we usually omit the $o(1)$ and assume that two $n\times n$ matrices can be multiplied using $O(n^{\omega})$ arithmetic operations.

We obtain the following result:

\begin{restatable}{theorem}{stAllHop}
\label{thm:s-t-alg}
There is an $\widetilde{O}(M n^{\omega})$ time algorithm
 that takes as input
a weighted directed graph $G = (V, E,w)$, with edge weights in $\{-M, \ldots, M\}$ and no negative cycles, and vertices $s, t\in V$, and outputs $d_{\leq
  h}(s, t)$ for all $1\leq h< n = |V|$. 
\end{restatable}

We also show (in the appendix) that the running time is conditionally optimal, even for DAGs with weights in $\{-1,1\}$, assuming that triangle detection in $n$-node graphs requires $n^{\omega-o(1)}$ time, a popular assumption in fine-grained complexity (see e.g., \cite{vsurvey}).

The proof of Theorem \ref{thm:s-t-alg} involves several different ingredients including the min-plus convolution of sequences of matrices. The crucial property we need from the ``at most $h$'' version of the problem is that there is at least one simple path achieving $d_{\le h}(s,v)$. This property also holds for the ``exactly $h$'' version of the problem on DAGs, so with minor changes to some technical details (mostly simplifications), \cref{thm:s-t-alg} also works for the ``exactly $h$'' version of the problem on DAGs with the same running time. In fact, all of our algorithms also work for this version of the problem. 

\paragraph{Single-Source All-Hops Shortest Paths.}
Given that the single-pair version of the problem admits an $\tilde{O}(Mn^\omega)$ time algorithm, the next natural question is whether the following Single-Source All-Hop Shortest Paths problem admits a similarly fast algorithm: Given a directed graph $G=(V,E,w)$ with integer weights in $\{-M,\ldots,M\}$, no negative weight cycles and a given source vertex $s\in V$, determine $d_{\leq h}(s,v)$, for every $h\in \{1,\ldots,n-1\}$ and every $v\in V$.

We present the first truly-subcubic time algorithm for this problem, for any $M\leq O(n^{0.128})$. 

\begin{restatable}{theorem}{SSAllHop}
\label{thm: single-source}
There is an $\widetilde{O}(Mn^{\omega+1/2})$ time algorithm that given a
directed graph $G = (V, E,w)$ with edge weights in $\{-M, \ldots M\}$ and no negative cycles, and a vertex $s\in V$, computes $d_{\leq h}(s,
v)$ for all $1\leq h\leq n-1$ and $v\in V$.
\end{restatable}

If $\omega=2$, the running time simplifies to $\widetilde{O}(Mn^{2.5})$.

One might wonder why this running time is higher than that for the single-pair version of the problem. As mentioned earlier, our algorithm achieving the above result can be adapted to also compute (in the same running time) the ``exactly $h$'' distances $d_h(s,v)$ between a given source $s$ and all targets $v$ in a DAG with weights in $\{-M,\ldots,M\}$. For this version of the problem, with $M=O(1)$, we obtain a matching conditional lower bound for the case when $\omega=2$:
\begin{restatable}{theorem}{SSLowerBound}
    \label{thm:ss-lower-bound}
        If one can compute $d_{h}(s, v)$ for all $v \in V, h \in [n]$ in an $n$-vertex directed acyclic graph $G=(V, E)$ with edge weights in $\{1, 2\}$ in $O(n^{2.5-\eps})$ time for $\eps>0$, then one can solve unweighted directed APSP in $O(n^{2.5-\eps})$ time, thus refuting the directed unweighted APSP hypothesis of \cite{DBLP:conf/icalp/ChanWX21}.
\end{restatable}

The choice of $\{1,2\}$ is somewhat arbitrary. Any choice of two different weights can work.

The above lower bound suggests that if one wanted to solve the Single Source All-Hop Shortest Paths problem in $O(n^{2.5-\eps})$ for $\eps > 0$, one needs to avoid algorithms that also solve the exact-hop version for DAGs. 

\paragraph{All-Pairs All-Hops Shortest Paths and distance oracles.}
By running Bellman-Ford from every vertex, one can compute $d_{\leq h}(u,v)$ for all pairs of vertices $u,v$ in a graph and all hop bounds $h=1,\ldots,n-1$ in $O(mn^2)\leq O(n^4)$ time.
A natural question is: {\em Can one obtain a sub-quartic running time?}

We provide a fine-grained negative answer under the Min-Plus Convolution hypothesis, which is a standard hypothesis in fine-grained complexity used in many prior works (see, e.g., \cite{CMWW19,KunnemannPS17,KoAd23}):

\begin{restatable}{theorem}{APLowerBound}
\label{thm:weighted-lower-bound}
    Under the Min-Plus Convolution hypothesis, computing $d_{\le h}(u, v)$ in an $n$-node weighted graph for every $u, v \in V, h \in [n]$ requires $n^{4-o(1)}$ time. 
\end{restatable}

As a simple corollary we obtain a lower bound for all-hops distance oracles. The corollary follows since the total number of queries one can make to the oracle is $O(n^3)$.

\begin{cor}
Under the Min-Plus Convolution hypothesis, there is no distance oracle that can preprocess a weighted graph in $O(n^{4-\eps})$ time for some $\eps>0$ and also answer $d_{\leq h}(u,v)$ queries for every triple $(u,v,h)$ in $O(n^{1-\delta})$ time for some $\delta>0$. 
\end{cor}

We complement our conditional lower bounds with distance oracles whose performance is {\bf tight}.
\begin{restatable}{theorem}{DOmn}
\label{thm:do-mn}
There exists an all-hops distance oracle for weighted graphs without negative cycles that has $\widetilde{O}(mn)$ preprocessing time and
$\widetilde{O}(n)$ query time and uses $\widetilde{O}(n^2)$ space. 
\end{restatable}

As the preprocessing time $\widetilde{O}(mn)$ is truly subquartic,  the $\widetilde{O}(n)$ query time cannot be improved under the Min-Plus Convolution hypothesis.

For dense graphs we obtain a further improvement via a reduction to the Min-Plus Product problem for which Williams \cite{Williams18} obtained an $n^3/2^{\Theta(\sqrt{\log n})}$ time algorithm.

\begin{restatable}{theorem}{DOMPP}
    There exists an all-hops distance oracle for weighted graphs without negative cycles that has $n^3/2^{\Theta(\sqrt{\log n})}$ preprocessing time and
$\widetilde{O}(n)$ query time and uses $\widetilde{O}(n^2)$ space. 
\end{restatable}

Finally, we consider the all-pairs version of all-hops shortest paths when the graph has bounded integer weights in $\{-M,\ldots,M\}$ and no negative cycles.
Our result for the single-source problem immediately implies that this all pairs version admits an $\widetilde{O}(Mn^{\omega+3/2})\leq O(Mn^{3.8714})$ time algorithm. We show that this running time can be improved substantially:

\begin{restatable}{theorem}{APAllHop}
    \label{thm:ap-upper-bound}
    There is an $\widetilde{O}(M^{0.5} n^{3.5})$ time algorithm
 that, given a
directed graph $G = (V, E)$ with edge weights in $\{-M, \ldots M\}$ and no negative cycles, computes $d_{\leq h}(u,
v)$ for all $1\leq h\leq n-1$ and $u, v\in V$.
\end{restatable}

Our approach also gives an improved construction of distance oracles for all-hops shortest paths for graphs with bounded integer weights.

\begin{restatable}{theorem}{DObounded}
    There exists an all-hops distance oracle for graphs with integer edge weights in $\{-M, \ldots, M\}$ and no negative cycles that has $\tO(M^{2/3} n^{(6+\omega)/3})$ preprocessing time and
$\widetilde{O}(n)$ query time and uses $\widetilde{O}(n^2)$ space. 
\end{restatable}

\subsection{Related results}

Only a few theoretical papers were previously written on the all-hops shortest paths problem and its variants. As mentioned, the classical Bellman-Ford algorithm \cite{Bellman58,Ford56} solves the single-source version of the problem in $O(mn)$ time. Gu{\'{e}}rin and Orda \cite{GuOr02} point this out and also obtain a corresponding lower bound for a \emph{path-comparison-based} algorithms, adapting a lower bound of Karger, Koller and Phillips \cite{KKP93} for the All-Pairs Shortest Paths (APSP) problem. They also consider a bottleneck version of the all-hops problem. Similar results are obtained by Cheng and Ansari \cite{ChAn04}. As also mentioned, Kociumaka and Polak \cite{KoAd23} showed that the $O(mn)$ bound is conditionally optimal under the APSP hypothesis; for graphs that are not so dense, their lower bound holds under the Min-Plus Convolution hypothesis.

Shortest paths with some bounds on the number of hops were used by several researchers as a means towards obtaining efficient shortest paths algorithms in the sequential, parallel or distributed settings. (See, e.g., \cite{UlYa91}, \cite{Cohen97}, \cite{AndoniSZ20}.) 

We next mention some results for the standard shortest path problem that we either use or borrow ideas from. We begin with results for general edge weights. APSP can be easily solved in $O(n^3)$ time using the Floyd-Warshall algorithm \cite{Floyd62,Warshall62}. It is also known that the problem is equivalent to the \emph{min-plus} product (MPP) of two $n\times n$ matrices. Williams \cite{Williams18} obtained an $n^3/2^{\Theta(\sqrt{\log n})}$-time algorithm for the MPP problem, and hence also for the APSP problem. (For a deterministic version, see Chan and Williams \cite{ChWi21}.) It is conjectured that there is no truly sub-cubic algorithm for the APSP problem, i.e., an algorithm whose running time is $O(n^{3-\eps})$ for some $\eps>0$, even if the edge weights are of polynomial size. This is known as the APSP hypothesis and is the basing hypothesis of many conditional lower bounds (see~\cite{RodittyZ11,focs10,vsurvey}). 

When edge weights are small integers, i.e., in the range $\{0,1,\ldots,M\}$ or $\{-M,\ldots,M\}$, where $M\ll n$, truly subcubic algorithms for the APSP problem can be obtained using fast matrix multiplication. 
The APSP problem for undirected graphs with edges weights in the range $\{0,1,\ldots,M\}$ can be solved in $\Ot(Mn^\omega)$ time~\cite{DBLP:journals/jcss/AlonGM97,shoshanzwick}. 
The APSP problem for directed graphs with edge weights in the range $\{-M,\ldots,M\}$ can be solved in $\Ot(M^{1/(4-\omega)}n^{2+1/(4-\omega)})$ or slightly faster using rectangular matrix multiplication (see Zwick \cite{Zwick02}).  For instance, using current bounds on rectangular matrix multiplication \cite{WXXZ24}, the running time is $O(M^{0.751596} n^{2.527661})$. If $\omega=2$, then the running time of the algorithm is $O(\sqrt{M}n^{2.5})$. 
This running time was conjectured to be optimal (see \cite{lincoln2020monochromatic,DBLP:conf/icalp/ChanWX21}).

 Yuster and Zwick \cite{YuZw05} showed that a weighted directed graph can be preprocessed in $\Ot(Mn^\omega)$ time after which each distance query can be answered in $O(n)$ time. Zwick \cite{Zwick99} and Chan, Vassilevska Williams and Xu \cite{DBLP:conf/icalp/ChanWX21} obtained algorithms for finding shortest paths that use a minimal number of hops.

%% file: prelim.tex
\section{Preliminaries}

For integer $n \in \mathbb{N}$, we use $[n]$ to denote $\{1, \ldots, n\}$. 

Let $G = (V, E, w)$ be an $n$-node $m$-edge directed graph with weight function $w$. We will assume $m \ge n - 1$ throughout the paper. The weighted adjacency matrix $W$ of $G$ is an $n \times n$ matrix where $W[u, v] = w(u, v)$ if $(u, v) \in E$, and $\infty$ otherwise. For $u, v \in V$, we use $d(u, v)$ to denote the length of the shortest path from $u$ to $v$. (Throughout the paper, length refers to the weighted length of a path, i.e., its length according to the weight function~$w$.) For $u, v \in V$ and an integer $h \ge 0$, we use $d_h(u, v)$ (resp. $d_{\leq h}(u, v)$) to denote the length of the shortest path from $u$ to $v$ that consists of exactly (resp. at most) $h$ hops (i.e., edges). For a path $p$, we use $|p|$ to denote its number of hops.

\begin{theorem}[\cite{Bellman58, Ford56}]
Given an $m$-edge directed weighted graph $G = (V, E, w)$, a vertex $s \in V$ and an integer $L \ge 1$, there exists an algorithm  $\BF{G}{s}{L}$ that computes $d_{h}(s, v)$ and $d_{\le h}(s, v)$ for all $v \in V$  and $1 \le h \le L$ in $O(mL)$ time.
\end{theorem}

\begin{definition}[min-plus product]
Given two matrices $A$ and $B$ of dimensions $m\times n$ and $n\times
p$, respectively, we define the min-plus product $C = A\star B$ by $C[i, j] = \min_{k\in [n]}\{A[i, k] + B[k, j]\}$ for
all $i\in [m]$ and $j\in [p]$. If $A$ is a square matrix, we use $A^q$ for $A$ to
the $q$th power under min-plus product, where $q \in\mathbb{N}$.
\end{definition}

We need the following two results on min-plus product between special matrices. 

\begin{lemma}[\cite{DBLP:journals/jcss/AlonGM97}]
\label{lem:small-mpp}
Let $A$ be an $n^a \times n^b$ matrix and $B$ be an $n^b \times n^c$ matrix, where all entries of $A$ and $B$ are from $\{-M, \ldots, M\} \cup \{\infty\}$. Then the
min-plus product $A\star B$ can be computed in $\widetilde{O}(Mn^{\omega(a, b, c)})$
time. 
\end{lemma}

\begin{lemma}[{\cite[Lemma B.2]{DBLP:conf/icalp/ChanWX21}}]
\label{lem:mpp-finite-B}
Let $A$ an $n^a\times n^b$ integer matrix and let $B$ an $n^b\times n^c$ integer matrix, where all finite entries of $B$ are bounded by $L$ in absolute value and the entries of $A$ are bounded by $\textrm{poly}(n)$. Then for any parameter $0 \le t \le a + b$, the min-plus product $A\star B$ can be computed in time 
\[
\tO(L n^{\omega(a+b-t, b, c)} + n^{a+c+t}).
\]
\end{lemma}

\begin{definition}[min-plus convolution]
Given two sequences $A$ and $B$ of length $n$, their min-plus convolution  $C = A\star B$ is a length $2n-1$ sequence where $C_i = \min_{1 \le k < i} \{A_k + B_{i - k}\}$ (out of boundary entries are treated as $\infty$).
\end{definition}

The min-plus convolution of two sequences of integers is a popular problem in fine-grained complexity (e.g. \cite{CMWW19,KunnemannPS17,JansenR23} ). Here, we also consider a version of the problem for two sequences of matrices, previously studied in \cite{GawrychowskiMW21}.

\begin{definition}[min-plus convolution of two matrix sequences]
    Given two sequences of $n \times n$ matrices $\cA =
    \langle A_{\ell_1}, \ldots, A_{r_1}\rangle$ and $\cB = \langle B_{\ell_2}, \ldots,
    B_{r_2}\rangle$, we define their min-plus-product-convolution $\cC = \cA \MinPlusConv \cB$ as a sequence of matrices $\langle C_{\ell_1+\ell_2}, \ldots, C_{r_1+r_2}\rangle$, where 
    \[
    C_{z}[i, j] = \min_{\substack{\ell_1 \le x \le r_1\\ \ell_2 \le y \le r_2 \\ x+y=z}} (A_x \star B_y)[i, j]
    \]
     for every $\ell_1 + \ell_2 \le z \le r_1 + r_2$ and $i, j \in [n]$.
\end{definition}

In \cite{GawrychowskiMW21}, it was shown that min-plus convolution between two length-$m$ sequences consisting of $n \times n$ matrices with arbitrarily large integer entries requires $n^3 m^{2-o(1)}$ time, under the Min-Plus Convolution hypothesis. If the matrices instead have small integer entries that are bounded in absolute values, we can solve the problem faster:

\begin{prop}
\label{prop:conv-min-plus}
    Given two sequences of $n \times n$ matrices $\cA =
    \langle A_{\ell_1}, \ldots, A_{r_1}\rangle$ and $\cB = \langle B_{\ell_2}, \ldots,
    B_{r_2}\rangle$ where the entries of the matrix elements of $\cA$ and
    $\cB$ are from $\{-M, \ldots, M\} \cup \{\infty\}$, we can compute $\cC = \cA
    \MinPlusConv \cB$ in $\widetilde{O}(M (r_1 - \ell_1 + 1 + r_2 - \ell_2 + 1) n^\omega)$ time.
\end{prop}
\begin{proof}
By appropriate padding and shifting, we can assume WLOG that $\cA = \langle A_{0}, \ldots, A_{m-1}\rangle$ and $\cB = \langle B_{0}, \ldots, B_{m-1}\rangle$, and all entries of the matrices are from $\{0, \ldots, 5M\}$. 

Define an $n \times n$  matrix $\mathbb{A}$ whose entries are polynomials in two variables $x, y$, where
\[
\mathbb{A}[i, j] = \sum_{0 \le \ell < m} x^\ell y^{A_\ell[i, j]},
\]
and 
\[
\mathbb{B}[i, j] = \sum_{0 \le \ell < m} x^\ell y^{B_\ell[i, j]}.
\]
Then the $(i, j)$-th entry of the product between $\mathbb{A}$ and $\mathbb{B}$ is 
\begin{align*}
    \sum_{1 \le k \le n}\left(\sum_{0 \le \ell < m} x^\ell y^{A_\ell[i, k]}\right)\left(\sum_{0 \le \ell < m} x^\ell y^{B_\ell[k, j]}\right) = \sum_{0 \le \ell < 2m-1}  \sum_{\substack{0 \le \ell_1 < m \\ 0 \le \ell_2 < m \\ \ell_1 + \ell_2 = \ell}} \sum_{1 \le k \le n} x^\ell y^{A_{\ell_1}[i, k] + B_{\ell_2}[k, j]}. 
\end{align*}
Hence, for every $i, j \in [n]$, $0 \le \ell < 2m-1$, $C_\ell[i, j]$ is exactly the degree of the minimum-degree nonzero monomial of coefficient of $x^\ell$ (the coefficient is a polynomial in $y$) in $(\mathbb{A} \mathbb{B})[i, j]$. Therefore, it suffices to compute the product between $\mathbb{A}$ and $\mathbb{B}$. As the maximum degree on $x$ of all entries is $O(m)$, and the maximum degree on $y$ of all entries is $O(n)$, $\mathbb{A} \mathbb{B}$ can be computed in $\tO(mMn^\omega)$ time (see e.g., \cite{DBLP:conf/stoc/ChiDX022}). 
\end{proof}

The relevance of min-plus convolutions of sequences of matrices to the all-hops shortest path problem is demonstrated by the following lemma.

\begin{lemma}
    Let $G=(V,E,w)$ be a weighted directed graph with $n=|V|$. Let $D_k$ be an $n\times n$ matrix such that $D_k[u,v]=d_{\le k}(u,v)$, for every $u,v\in V$. Then,
    \[\left\langle D_0,D_1,\ldots,D_{2k} \right\rangle \EQ \left\langle D_0,D_1,\ldots,D_k \right\rangle \MinPlusConv \left\langle D_0,D_1,\ldots,D_k \right\rangle \;. \]
\end{lemma}

We  need the following standard result.

\begin{lemma}\label{lemma:hitting_set}
(Hitting set lemma.)
Let $S = \{S_1, \ldots, S_N\}$ be a family of $N$ subsets of $[L]$ such that 
$|S_i| > k$ for every $i\in[N]$. Then a uniformly random subset $H\sub
[L]$ of size at least $C(L/k)\ln N$, where $C\geq 1$, hits every set
in $S$ with high probability: for all $1\leq i\leq N$, $H\cap
S_i\neq\varnothing$ with probability at least $1-1/N^{C-1}$.
  
\end{lemma}

%% file: s-t-alg-new.tex
\section{Single-Pair All-Hops Distances}\label{sec:st-version}

In this section, we prove \cref{thm:s-t-alg}, which we recall below:
\stAllHop*

For any integer $k\ge 1$, we describe an 
$\widetilde{O}(k M n^{\omega+2/k})$-time algorithm  that, given an $n$-vertex directed graph $G = (V,E,w)$, where $w:E\to\{-M,\ldots,M\}$, with no negative cycles, and vertices $s, t\in V$, computes $d_{\leq h}(s, t)$ for all $1 \le h < n$. 

The algorithm performs $k$ iterations. We let $r$ be the index of the current iteration, starting from~$0$. In each iteration we choose a random sample $S_r$ of the vertices where $S_0=V$ and where $S_0\supseteq S_1 \supseteq S_2\supseteq \cdots \supseteq S_{k}$. We also make sure that $s,t\in S_r$, for every $0\le r\le k$. The size of $S_r$ is $C n^{1-r/k}\log n$, for a sufficiently large constant $C>0$ (except $S_0$, which has size $n$). In the $r$-th iteration we compute $d_{\le h}(u,v)$ for every $u,v\in S_r$ and $h\le n^{r/k}$. In the $0$-th iteration, with $S_0=V$, we need to compute $d_{\le 1}(u,v)$, for every $u,v\in V$, which is just the weighted adjacency matrix of the graph (except the diagonal terms, which are all $0$). After the $k$-th iteration we have $d_{\le h}(s,t)$ for every $1\le h<n$, since $s,t\in S_k$, as required.

Next, we consider iteration $r > 0$. For simplicity, let $H := n^{(r-1)/k}$. 

For every $0 \le i \le L$, where $L = \lfloor \log(n^{r/k}) \rfloor$,  we can define a sequence of matrices 
\[\cD_i := \langle d_{\leq
  j}(S_{r-1}, S_{r-1})\rangle_{j=2^i-H/2}^{2^i+H/2}, 
\] 
where $d_{\le j}(S_{r-1}, S_{r-1})$ denotes an $S_{r-1} \times S_{r-1}$ matrix whose $(u, v)$-th entry is $d_{\le j}(u, v)$ (if $j < 0$, then $d_{\le j}(u, v) = \infty$ for every $u, v \in V$). We can easily compute $\cD_0$ because we have computed $d_{\le h}(s_1, s_2)$ for $s_1, s_2 \in S_{r-1}$ and $h \in [H]$ in the previous iteration. For $1 \le i \le L$, we will show that in order to compute $\cD_{i}$, it suffices to compute the min-plus convolution of
$\cD_{i-1}$ with itself.

For $s_1, s_2 \in S_{r-1}$ and $2^i-H/2\leq h\leq 2^i+H/2$, let $p$ be a simple path achieving $d_{\le h}(s_1, s_2)$ (we can assume $p$ is simple because $G$ has no negative cycles). By Lemma \ref{lemma:hitting_set},
w.h.p., there is a vertex $x\in S_{r-1}$ that hits the path $p$ and
splits it into two subpaths $p_1 = (s_1, \ldots, x)$ and $p_2 = (x,
\ldots, s_2)$ such that $|p|/2 - H/4 \le |p_1|, |p_2| \le |p|/2 + H / 4$ (a corner case is when $|p| < H/2$, in which case we can simply take $x = s_1$ without applying \cref{lemma:hitting_set}).

By taking $j = \max\{h / 2 - H/4, |p_1|\}$, we get the following properties\footnote{Here, we have to do some careful calculation. For the algorithm for DAGs that computes $d_{h}(s_1, s_t)$, we have the guarantee that $h/2 - H/4 \le |p_1|$, so the value of $j$ will simply become $|p_i|$, and much of the calculation can be simplified. }: 
\begin{itemize}
    \item $|p_1| \le j$. This is clearly true. 
    \item $|p_2| \le h - j$. This is because $h - j \ge |p| - |p_1| \ge |p_2|$.
    \item $j \in [2^{i-1} - H/2, 2^{i-1} + H/2]$. First, we have $j \ge h / 2 - H/4 \ge (2^i - H/2) / 2 - H/4 = 2^{i-1} - H/2$. Also, $j = \max\{h / 2 - H/4, |p_1|\} \le  \max\{h / 2 - H/4, |p|/2 + H/4\} \le h / 2 + H/4 \le (2^i + H/2) / 2 + H/4 = 2^{i-1} + H/2$. 
    \item $h - j \in [2^{i-1} - H/2, 2^{i-1} + H/2]$. First, $h - j \le h - (h/2 - H/4) \le 2^{i-1} + H/2$. Also, $h - j = h - \max\{h/2 - H/4, |p_1|\} \ge h - \max\{h/2 - H/4, |p|/2 + H/4\}\ge h - (h / 2 + H/4) \ge 2^{i-1} - H/2$.
\end{itemize}
This implies 
\[
d_{\le h}(s_1, s_2) = \min_{\substack{2^{i-1}-H/2 \le j \le 2^{i-1}+H/2\\ 2^{i-1}-H/2 \le h-j \le 2^{i-1}+H/2}} \min_{x \in S_{r-1}} \left\{d_{\le j}(s_1, x) + d_{\le h - j}(x, s_2) \right\},
\]
which is exactly $(\cD_{i-1} \MinPlusConv \cD_{i-1})_h[s_1, s_2]$.

\begin{claim}\label{cl:s-t-alg-1}

Given $\cD_{i-1}$, we can compute $\cD_{i}$ in
$\widetilde{O}(M n^{\omega+1/k})$ time, where $1\leq i\leq
L$.
\end{claim}

\begin{proof}

By previous discussion, it suffices to compute $\cD_{i-1} \MinPlusConv \cD_{i-1}$. By  \cref{prop:conv-min-plus},
since $\cD_{i-1}$ is a length-$O(n^{(r-1)/k})$ sequence of matrices of dimension
$|S_{r-1}|\times |S_{r-1}| = \tO(n^{1-(r-1)/k}) \times \tO(n^{1-(r-1)/k})$ whose finite entries are at most $O(M n^{r/k})$, we can compute
$\cD_{i-1} \MinPlusConv \cD_{i-1}$ in 
\[\widetilde{O}(M n^{r/k}n^{(r-1)/k}(n^{1-(r-1)/k})^\omega) = \tO(M n^{\omega + 1/ k - (\omega - 2) \cdot (r-1)/k}) =  \tO(M n^{\omega + 1/ k})
\]
time. 
\end{proof}

Next, we use $\cD_0, \ldots, \cD_{L}$ to compute $d_{\leq h}(S_r,
S_r)$ for all $h\in [n^{r/k}]$. We do this in $L+1$ iterations,
starting from iteration $0$. In iteration $i$, we use $\cD_i$ and
$d_{\leq h}(S_r, S_{r-1})$ for $1\leq h\leq 2^{i}$ (which
are already computed prior to the start of the iteration)
to compute $d_{\leq h}(S_r, S_{r-1})$ for all $1\leq h\leq
2^{i + 1}$. If $2^{i+1} \le n^{(r-1)/k}$, we can simply use distances computed in previous iterations, so we assume $2^{i+1} > n^{(r-1)/k}$. For $s_1 \in S_r$, $s_2 \in S_{r-1}$, and $1 \leq h\leq 2^{i+1}$, let $p$ be a simple shortest path achieving $d_{\le h}(s_1, s_2)$. If $|p| \le 2^i $, then we can use $d_{\le 2^i} (s_1, s_2)$ (which was known before the iteration) to compute $d_{\le h}(s_1, s_2)$, so we can assume $|p| > 2^i$ without loss of generality. By Lemma~\ref{lemma:hitting_set},
w.h.p., there is a vertex $x\in S_{r-1}$ that hits the path $p$ and
splits it into two subpaths $p_1 = (s_1, \ldots, x)$ and $p_2 = (x,
\ldots, s_2)$ such that $2^{i}\leq
|p_2|\leq 2^{i}+n^{(r-1)/k}/2$ holds, which further implies $|p_1| \le h - 2^i \le 2^i$.  

Given this, we can compute $d_{\leq
  h}(S_r, S_{r-1})$ for $1 \leq h\leq 2^{i + 1}$ by
computing: $$\langle d_{\leq j}(S_r,
S_{r-1})\rangle_{j=1}^{2^{i}}\MinPlusConv \langle d_{\leq j}(S_{r-1},
S_{r-1})\rangle_{j=2^{i}}^{2^{i}+n^{(r-1)/k}/2}.$$

At the end of $L+1$ rounds of iterations of this inner loop, we have
computed the sequence $\langle d_{\leq j}(S_r,
S_r)\rangle_{j=1}^{n^{r/k}}$.

\begin{claim}\label{cl:s-t-alg-2}

For any $i \in [1, L]$, given $\cD_i$ and $\langle d_{\leq j}(S_r,
S_{r-1}) \rangle_{j=1}^{2^{i}}$, we can compute $\langle
d_{\leq j}(S_r, S_{r-1}) \rangle_{j=1}^{2^{i+1}}$ in
$\widetilde{O}(M n^{\omega+2/k})$ time.
  
\end{claim}

\begin{proof}

By previous discussion, it suffices to compute $\langle d_{\leq j}(S_r,
S_{r-1}) \rangle_{j=1}^{2^{i}} \MinPlusConv \cD_i$. By
\cref{prop:conv-min-plus}, since $\langle d_{\leq j}(S_r,
S_{r-1}) \rangle_{j=1}^{2^{i}}$ is a length-$O(n^{r/k})$ sequence of
matrices of dimension $|S_r|\times |S_{r-1}| \le \tO(n^{1-(r-1)/k}) \times \tO(n^{1-(r-1)/k})$ whose finite entries are at most $O(Mn^{r/k})$ and $\cD_i$ is a length-$O(n^{r/k})$ sequence of matrices of dimension
$\tO(n^{1-(r-1)/k}) \times \tO(n^{1-(r-1)/k})$ whose finite entries are at most $O(Mn^{r/k})$, we can compute their min-plus-product-convolution in 
\[
\widetilde{O}(M n^{r/k}n^{r/k}(n^{1-(r-1)/k})^\omega) =
\widetilde{O}(M n^{\omega + 2/k - (\omega - 2) \cdot (r-1)/k}) = \widetilde{O}(M n^{\omega + 2/k})
\]
time. 
\end{proof}

By \cref{cl:s-t-alg-1,cl:s-t-alg-2}, each iteration $r > 0$ of the algorithm takes
$\widetilde{O}(M n^{\omega+1/k}+ n^{\omega+2/k}) =
\widetilde{O}(M n^{\omega+2/k})$ time. As there are $k$ iterations in total, the overall running time of the algorithm becomes $\widetilde{O}(k M n^{\omega+2/k})$. 

Finally, setting $k = \log n$ gives the desired $\tO(M n^{\omega})$ running time. 

%% file: single-source.tex
\section{Single-Source All-Hops Distances}

In this section, we study Single-Source All-Hops Distances. We begin by proving \cref{thm: single-source}, which we recall below:
\SSAllHop*

\begin{proof}

Similar to \cref{sec:st-version}, we describe an algorithm for any integer $k \ge 1$, and the algorithm will perform~$k$ iterations. We let $r$ be the index of the current iteration, starting from~$0$. In each iteration we choose a random sample $S_r$ of the vertices where $S_k=V$ and where $S_0\subseteq S_1 \subseteq S_2\subseteq \cdots \subseteq S_{k}$ (we already deviate from \cref{sec:st-version} in this step). We also make sure that $s\in S_r$, for every $0\le r\le k$.  The size of $S_r$ is $C n^{r/k}\log n$, for a sufficiently large constant $C>0$ (except $S_k$, which has size $n$). In the $r$-th iteration we compute $d_{\le h}(s,v)$ for every $v\in S_r$ and $h \in [n]$. At the end of the $k$-th iteration, we will have $d_{\le h}(s, v)$ for every $v \in S_k = V$ and $h \in [n]$, as desired.

For each iteration, we describe two algorithms. We can choose to use the more efficient one, depending on how big $r$ is. 

\paragraph{First algorithm. } For the first algorithm, we simply call the Single-Pair All-Hops Distances algorithm from \cref{thm:s-t-alg} $|S_r|$ times, and each time, we set $t$ to be a different vertex in $S_r$. Clearly, this algorithm is able to compute $d_{\le h}(s,v)$ for every $v\in S_r$ and $h \in [n]$. The running time is $\tO(|S_r| \cdot n^\omega) = \tO(n^{\omega + r / k}$). 

\paragraph{Second algorithm. }

The second algorithm works for $r > 0$, and has two main steps. 
\begin{enumerate}[label=(\arabic*)]
\item \label{item:SS-algo:step1} Compute the distances
$d_h(u, v)$ for all $h\in [n^{1-(r-1)/k}]$, $u\in S_{k-1}$, and $v\in
S_k$.

Let $W$ be the weighted adjacency matrix of $G$, and let $A_h$ for $h \in [n]$ be $d_h(S_{k-1}, V)$, the matrix indexed by $S_{k-1} \times V$ whose $(u, v)$ entry is $d_h(u, v)$. To complete this step, it suffices to compute $A_h$ for all $h\in [n^{1-(r-1)/k}]$.

Notice that $A_1 = W[S_{k-1}, V]$, and
for $1 < h\leq b$, $A_h = A_{h-1}\star W$; therefore, we can compute
$A_h$ for all $h\in [n^{1-(r-1)/k}]$ via doing $O(n^{1-(r-1)/k})$ min-plus products between
the matrix $A_h$ of dimension $\tO(n^{(r-1)/k})\times n$ and the matrix $W$ of
dimension $n \times n$, whose finite entries are both bounded by $O(n^{1-(r-1)/k}M)$ in absolute value. Thus, by \cref{lem:small-mpp}, we can compute $A_{h-1}\star W$ in $O(n^{1-(r-1)/k}Mn^\omega)$ time (it can be improved using rectangular matrix multiplication, but the running time is still essentially $O(n^{1-(r-1)/k}Mn^2)$ when $\omega = 2$), which suggests that this approach might be wasteful:
in this much time, we could have done a min-plus product
between matrices of dimension $n\times n$ with the same entry bounds. To utilize this intuition, we stack multiple $A_h$'s into a single  matrix: to compute $A_h$ for all $h\in [n^{1-(r-1)/k}]$, rather
than do $O(n^{1-(r-1)/k})$ min-plus products between $A_h$ and $W$ for one $A_h$
at a time, we now do only $O(\log n)$ min-plus products between a
matrix consisting of several $A_h$'s stacked together and $W$.

We first do repeated min-plus squaring to compute the following
matrices:
\[W, W^2, W^4, W^8, \ldots, W^{n^{1-(r-1)/k}/2}.
\]
Then we set a variable $\cA = A_1 = W[S_{k-1}, V]$, and for $i$ from $0$ to $\log (n^{1-(r-1)/k}/2)$, we
compute $\cB = \cA \star W^{2^i}$ and set $\cA = \mathsf{Stack}(\cA, \cB)$, where $\mathsf{Stack}(\cA, \cB)$ is a matrix where $\cA$ is on top of~$\cB$. If
$\cA$ consists of the matrices $A_1, \ldots, A_{2^i}$ stacked together, then
$\cB = \cA\star W^{2^i}$ consists of the matrices $A_{2^i+1}, A_{2^i+2},
\ldots, A_{2^{i+1}}$ stacked together, and by stacking $\cA$ on top of $\cB$,
we obtain the matrix consisting of the matrices $A_1, \ldots, A_{2^{i+1}}$
stacked together. Upon exiting the for-loop, $A$ contains the matrices
$A_1, \ldots, A_{n^{1-(r-1)/k}}$ stacked together, which together contain the
desired distances $d_h(u, v)$ for all $h\in [n^{1-(r-1)/k}]$, $u\in S_{k-1}$ and
$v\in S_k$.

We next analyze the running time of this step. To
compute $W, W^2, W^4, \ldots, W^{n^{1-(r-1)/k}/2}$, we do $O(\log n)$ min-plus
products between $n \times n$ matrices whose finite entries are bounded by $O(Mn^{1-(r-1)/k})$ in absolute value; then, to compute $A_h$ for all $h\in [n^{1-(r-1)/k}]$, we
do $O(\log n)$ min-plus products between an $O(n^{1-(r-1)/k} \cdot |S_{k-1}|) \times n = \tO(n) \times n$ matrix and an $n \times n$ matrix, both with finite entries bounded by $O(M n^{1-(r-1)/k} )$ in absolute value. By Lemma
\ref{lemma:hitting_set}, each of these min-plus products takes
$\tO(Mn^{1-(r-1)/k}n^\omega)$ time. Hence, the running time of this step is
$\tO(Mn^{\omega + 1-(r-1)/k})$. 

\item \label{item:SS-algo:step2} Use the distances computed in the previous step, together with the distances
$d_{\leq h}(s, u)$ for all $h\in [n]$ and $u\in S_{k-1}$ (which have
been computed in the previous iteration), to compute the
distances $d_{\le h}(s, v)$ for all $h\in [n]$ and $v\in S_r$.

Observe that we have already computed in step 2 the distances $d_h(s, v)$
for all $h\in [n^{1-(r-1)/k}]$ and $v\in S_r$ (as $s\in S_{r-1}$), from
which we can readily obtain $d_{\leq h}(s, v)$ for all $h\in [n^{1-(r-1)/k}]$ and
$v\in S_r$. Define the matrix $A$ of dimension $n\times |S_{r-1}|$ by
$A[h, u] = d_h(s, u)$, where $h\in [n]$ and $u\in S_{r-1}$, and define
the matrix $B$ of dimension $|S_{r-1}| \times (n^{1-(r-1)/k} \cdot |S_r|)$ by $B[u, (v, h)] = d_h(u,
v)$, where $u \in S_{r-1}$, and $v \in S_r$ and $h\in [n^{1-(r-1)/k}]$. 
To compute the
remaining distances $d_{\leq h}(s, v)$ for all $n^{1-(r-1)/k}+1\leq h\leq n$ and
$v\in S_r$, we will show that it suffices to compute the matrix $C = A\star B$,
from which we can efficiently extract these distances. For every $v \in S_r$ and $n^{1-(r-1)/k}+1\leq h\leq n$, we can fix one simple shortest path $p$ achieving $d_{\le h}(s, v)$. If $p$ uses at most $n^{1-(r-1)/k}$ hops, then it has already been captured by the distances computed in the previous step. Otherwise, by \cref{lemma:hitting_set}, w.h.p., there must be a vertex in $S_{r-1}$ that belongs to the last $n^{1-(r-1)/k}$ vertices on $p$. From this, for all $n^{1-(r-1)/k}+1\leq h\leq n$ and $v\in
S_r$, we obtain the following:
\begin{align*}
d_{\leq h}(s, v) &= \min\left\{d_{\le n^{1-(r-1)/k}}(s, v), \min_{h'\in [n^{1-(r-1)/k}], u\in S_{r-1}}\left\{d_{\leq h-h'}(s, u)+d_{h'}(u, v)\right\}\right\}\\
&=\min\left\{d_{\le n^{1-(r-1)/k}}(s, v), \min_{h'\in [n^{1-(r-1)/k}], u \in S_{r-1}} \left\{A[h-h', u] + B[u, (v, h')]\right\}\right\}\\
&= \min\left\{d_{\le n^{1-(r-1)/k}}(s, v), \min_{h' \in [n^{1-(r-1)/k}]} \left\{C[h-h',(v, h')]\right\}\right\}.
\end{align*}

The bottleneck of this step is to compute the min-plus product $C = A\star B$. Note that $A$ is an $n \times \tO(n^{(r-1)/k})$ matrix, $B$ is an $\tO(n^{(r-1)/k}) \times \tO(n^{1+1/k})$ matrix, and all finite entries of $B$ are bounded by $O(M n^{1-(r-1)/k})$ in absolute value, so by \cref{lem:mpp-finite-B}, we can compute $A \star B$ in time
\[
\tO(M n^{1-(r-1)/k} \cdot n^{\omega(1+(r-1)/k-t, (r-1)/k, 1+1/k)} + n^{2+1/k+t}).
\]
for any $t \ge 0$ (technically, we also need $t \le 1+(r-1)/k$. However, when $t > 1+(r-1)/k$, the second term in the running time becomes greater than $n^{3+r/k}$, slower than brute-force. Hence, we can ignore the upper bound on $t$ for simplicity). We can use 
\[
\omega(1+(r-1)/k-t, (r-1)/k, 1+1/k) \le 2 - t - r/k + ((r-1)/k)\cdot \omega
\]
(by splitting the matrix multiplication to several matrix multiplications between square matrices of dimension $n^{(r-1)/k} \times n^{(r-1)/k}$), to further bound the running time by 
\begin{align*}
\tO(M n^{1-(r-1)/k} \cdot n^{2 - t - r/k + ((r-1)/k)\cdot \omega} + n^{2+1/k+t}) & = \tO(M n^{3-t+1/k + (\omega - 2)(r-1)/k} + n^{2+1/k+t})\\
& = \tO(M n^{\omega + 1-t+1/k} + n^{2+1/k+t}).
\end{align*}

We can set $t$ so that $n^t = \sqrt{M} n^{(\omega - 1) / 2}$, so that the running time becomes $\tO(\sqrt{M} n^{(3+\omega) / 2 + 1/k})$. 
\end{enumerate}
Overall, the running time of the second algorithm is 
\[
\tO\left(Mn^{\omega + 1-(r-1)/k} + \sqrt{M} n^{(3+\omega) / 2 + 1/k} \right). 
\]

\paragraph{Final running time. } For $r \le k / 2$, we can choose the first algorithm, which has running time $\tO(n^{\omega + r / k}) = \tO(n^{\omega + 1/2})$. For $r > k/2$, we choose the second algorithm, which has running time 
\[
\tO\left(Mn^{\omega + 1-(r-1)/k} + \sqrt{M} n^{(3+\omega) / 2 + 1/k} \right) = \tO\left(Mn^{\omega + 1/2 + 1/k}\right).
\]
We also need to pay an $O(k)$ factor for the number of iterations, so the overall running time is $\tO\left(k Mn^{\omega + 1/2 + 1/k}\right)$. By setting $k = \log n$, we achieve the desired $\tO\left(Mn^{\omega + 1/2}\right)$ time. 
\end{proof}  

As discussed in \cref{sec:results}, \cref{thm: single-source} can also compute $d_h(s, v)$ for all $v \in V$ and $h \in [n]$ in directed acyclic graphs, in the same running time. The running time becomes $\tO(n^{2.5})$ for $M = O(1)$ and $\omega = 2$. Next, show that this running time is conditionally optimal by proving \cref{thm:ss-lower-bound}, which we recall below:
\SSLowerBound*

The reduction uses the following gadget (see Figure \ref{fig:treegadget}): 
\begin{claim}
\label{cl:ss-lower-bound}
    For every integer $\ell \ge 1$, there is a graph (in fact a tree) with $O(\ell \cdot 2^\ell)$ vertices, including a sink vertex $v$ and $2^\ell$ source vertices $u_1, \ldots, u_{2^\ell}$, so that 
    \begin{itemize}
        \item For every $i$, there is a unique path from $u_i$ to $v$.
        \item The number of hops on this path is $2^{\ell} - 1$.
        \item The total edge weight on this path is $i+2^{\ell} - 2$. 
    \end{itemize}
\end{claim}
\begin{proof}

\begin{figure}[ht]\centering
\includegraphics[width=0.4\textwidth]{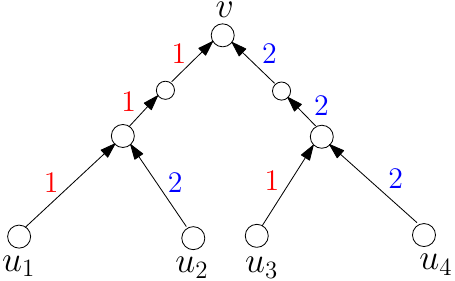}
\caption{A depiction of our gadget for $\ell=2$. Note that the distance between $u_i$ and $v$ is $2+i$ for every $i=1,2,3,4$. The number of hops on each $u_i$ to $v$ path is $3$.}
\label{fig:treegadget}
\end{figure}

We construct the graph by replacing tree edges of a complete binary tree with $2^\ell$ leaves with chains. The leaves of the binary tree are labeled by $u_1, \ldots, u_{2^\ell}$, and the root of the tree is labeled by $v$. 

Let the height of a vertex in the tree be its distance to the nearest leaf vertex. For an edge between a height $i+1$ vertex and a height $i$ vertex for $i \ge 0$, we replace it with a chain of $2^{i}$ edges. If the height $i$ vertex is the left child of the height $i+1$ vertex, all edges on this chain have weight $1$; otherwise, all edges on this chain have weight $2$. 

The claim can be verified straightforwardly, say by induction. 
\end{proof}

\begin{proof}[Proof of \cref{thm:ss-lower-bound}]

It was shown in \cite{Zwick02} one can solve  unweighted directed APSP by solving $\tO(1)$ instances of min-plus products  an $n \times (n/x)$ matrix $A$ and an $(n/x) \times n$ matrix $B$ with entries in $[x]$, for various values of $x \in [n]$ (formulated this way in \cite{DBLP:conf/icalp/ChanWX21}).

To compute the min-plus product between an $n \times (n/x)$ matrix $A$ and an $(n/x) \times n$ matrix $B$ with entries in $[x]$, we create the following graph (we assume $x = 2^\ell$ for some integer $\ell$ WLOG):
\begin{itemize}
    \item There are $n$ vertices $a_1, a_2, \ldots, a_n$ and $b_1, b_2, \ldots, b_n$. For every $1 \le i < n$, there is an edge of weight $1$ from $a_i$ to $a_{i+1}$. 
    \item For every $k \in [n/x]$, we use two copies of the gadgets from \cref{cl:ss-lower-bound}. We label the source vertices of the first copy by $c_{k,1}, \ldots, c_{k, x}$, and the source vertices of the second copy by $c'_{k, 1}, \ldots, c'_{k,x}$. We reverse all edges of the second copy. Finally, we identify the two sink vertices of the two copies. 
    \item For every $i, k$, we add an edge of weight $1$ from $a_i$ to $c_{k, A[i, k]}$. For every $k, j$, we add an edge of weight $1$ from $c'_{B[k, j]}$ to $b_j$. 
\end{itemize}
Let $C = A \star B$. Then we claim $C[i, j] = d_{i-1+2x}(s, b_j)-(i-3+2x)$. This is because the only type of path from $a_1$ to $b_j$ using $i-1+2x$ edges has the form $a_1 \rightarrow a_i \rightarrow c_{k, A[i, k]} \rightarrow c_{k, B[k, j]}' \rightarrow b_j$ for some choice of $k$. The total edge weights of such a path will be $i-3+2x + A[i, k] + B[k, j]$. Therefore, the shortest such path will choose the $k$ that minimizes $A[i, k] + B[k, j]$, so the claim follows. 
\end{proof}

%% file: all-pair.tex
\section{All-Pairs All-Hops Distances}

In this section, we prove \cref{thm:ap-upper-bound}:

\APAllHop*

\begin{proof}
    The high-level structure of the algorithm is as follows: it computes pairwise distances for hop lengths up to $(3/2)^k$ in the $k$-th iteration, and then use answers computed in previous iterations to help the next iteration. 

    Let $K = (3/2)^{k-1}$. Suppose we have computed $d_{\le h}(u, v)$ for $u, v \in V$ and $h \le K$, then we aim to compute $d_{\le h}(u, v)$ for $u, v \in V$ and $h \le 3K/2$. 
    
    Let $S_k \subseteq V$ be a random subset of vertices of size $O\left(\frac{n \log n}{K}\right)$. For any $u, v \in V$ and $K < h \le 3K/2$, fix a simple shortest path $p$ from $u$ to $v$ using at most $h$ hops. If $|p| \le K$, then $d_h(u, v) = d_{K}(u, v)$, which has been computed in previous iterations. Otherwise, with high probability, there exists a vertex $x \in S_k$ that splits $p$ into two subpaths $p_1$ (from $u$ to $x$) and $p_2$ (from $x$ to $v$), so that $|p_1|, |p_2| \le K$. Therefore, it suffices to compute 
    \[
    \min_{\substack{x \in S_k \\ h - K \le h' \le K}} \left\{ d_{\le h'}(u, x) + d_{\le h - h'}(x, v)\right\}.
    \]
    Therefore, in order to compute the whole sequence $\left\langle d_{\le h}(u, v)\right\rangle_{h = K + 1}^{3K/2}$, it suffices to compute the min-plus convolution between $\left\langle d_{\le h}(u, x)\right\rangle_{h=1}^{K}$ and $\left\langle d_{\le h}(x, v)\right\rangle_{h=1}^{K}$ for every $x \in S_k$. Observe that both sequences are non-increasing, and have integer entries bounded by $O(K M)$. Therefore, by \cite[Theorem 7]{BDP24}, each min-plus convolution can be computed in $\tO(K \cdot \sqrt{KM}) = \tO(M^{0.5} K^{1.5})$ time. Summing over all $u, v \in V$ and $x \in S_k$, the total running time is $\tO(M^{0.5} K^{0.5} n^3)$. 

    The largest value of $K$ is $\Theta(n)$, and the number of iterations we need to perform is $O(\log n)$, so the overall running time is $\tO(M^{0.5}n^{3.5})$. 
\end{proof}

%% file: distoracles.tex
\section{All-Hops Distance Oracles}

In this section, we study All-Hops Distance Oracles, where given an $n$-vertex $m$-edge weighted directed graph $G = (V, E, w)$ to preprocess, the oracle needs to answer
queries of the form $(u, v, h)$, where $u, v\in V$ and $h\in [n]$ with the distance $d_{\leq h}(u, v)$. 

We first give two simple constructions that have $O(1)$ query
time (at the cost of a large preprocessing time).

\begin{prop}
There exists an all-hops distance oracle for weighted graphs without negative cycles that has $n^4 / 2^{\Theta(\sqrt{\log n})}$ preprocessing time and $O(1)$
query time and uses $O(n^3)$ space.
\end{prop}

\begin{proof}
Given a directed graph $G = (V,
E)$ with arbitrary edge weights, we compute the matrices $W^i$ for all
$i\in [n]$ via $n-1$ min-plus products (where $W$ is the weighted
adjacency matrix of $G$). These matrices together contain the
distances $d_h(u, v)$ for all $u, v\in V$ and $h\in [n]$, from which
we can readily obtain the desired distances $d_{\leq h}(u, v)$ for all
$u, v\in V$ and $h\in [n]$. Each min-plus product can be solved in $n^3 / 2^{\Theta(\sqrt{\log n})}$~\cite{Williams18}, so the preprocessing time is $n^4 / 2^{\Theta(\sqrt{\log n})}$
as we do $O(n)$ min-plus products. The space used is $O(n^3)$, and the query time is $O(1)$.
\end{proof}

Unless the graph $G$ is extremely dense, the following construction
has a better preprocessing time.

\begin{prop}
There exists an all-hops distance oracle for weighted graphs without negative cycles that has $O(mn^2)$ preprocessing and $O(1)$ query
time and uses $O(n^3)$ space.
\end{prop}

\begin{proof}
Given a directed graph $G = (V,
E)$ with arbitrary edge weights, we call $\BF{G}{u}{n}$ for every
vertex $u\in V$, which gives us the distances $d_{\leq h}(u, v)$ for
all $u, v\in V$ and $h\in [n]$. The preprocessing time is $O(mn^2)$
as we run an $O(mn)$-time algorithm on $n$ vertices, the space used is
$O(n^3)$, and the query time is $O(1)$.
\end{proof}

We next present a construction that has $\widetilde{O}(mn)$ preprocessing time and $\widetilde{O}(n)$ query time.
\DOmn*

\begin{proof}
Given a directed graph $G = (V,
E)$ with arbitrary edge weights, we proceed in $\log n$ rounds of
iterations, starting with iteration $0$. In iteration $i$, we first
sample uniformly at random a set of vertices $S_i \subseteq V$ of size
$O(n\log n/2^i)$. Then, for every $s\in S_i$, we call $\BF{G}{s}{
2^{i+1}}$ and $\BF{R(G)}{s}{2^{i+1}}$, where we use $R(G)$ to denote a copy of $G$ where the direction of every edge is reversed,
which together give us the distances $d_{\leq h}(s, u)$ and $d_{\leq
h}(u, s)$ for all $u\in V$ and $h\in [2^{i+1}]$.

Next, we discuss how to answer a query $(u, v, h)$, where $u, v\in V$ and $h\in [n]$.
Let $i^*$ be the unique integer satisfying $2^{i^*}\leq h <
2^{i^*+1}$. 
 The oracle computes
\[d_{\leq h}(u, v) = \min_{\substack{i\in \{0,\ldots,i^*\}, s\in S_i \\ 
h'\in [\min\{h, 2^{i+1}\}]}} \left\{ d_{\leq h'}(u, s)+d_{\leq \min\{h-h', 2^{i+1}\}}(s, v)\right\},
\]

Suppose that $d_{\leq h}(u,v)$ is achieved by a simple $b$-hop path $P$ for some $b\leq h$ and let $i$ be the unique integer satisfying $2^i\leq b < 2^{i+1}$.
By Lemma \ref{lemma:hitting_set}, with high probability,
there is a vertex in $S_i$ that lies on $P$, and correctness follows.

Summing over all iterations of the outer loop, the preprocessing time
is
\[O\left(\sum_{i = 0}^{\log n}\frac{n\log n}{2^i}2^{i+1}m\right) = \widetilde{O}(mn),
\]
and the space used is
\[O\left(\sum_{i = 0}^{\log n}\frac{n\log n}{2^i}n\right) = \widetilde{O}(n^2).
\]
The query time is, within constant factors,
\[\sum_{i\leq i^*} 2^{i+1}n\log n/2^i\leq O(n\log^2 n)= \widetilde{O}(n).\]

\end{proof}

We next give a construction that has $n^{3} / 2^{\Theta(\sqrt{\log n})}$ preprocessing
time and $\widetilde{O}(n)$ query time. This construction improves the
preprocessing time of the construction given in Theorem
\ref{thm:do-mn} in the setting where the graph is dense. The basic idea
is to adapt the Yuster-Zwick distance oracle \cite{YuZw05} to the all-hops distance case.

\DOMPP*

\begin{proof}

Given a weighted directed graph $G = (V, E, w)$,  we
first initialize a variable $S_0 = V$, and, for $i = 1, \ldots,
\log_{3/2} n$, we sample a uniformly random subset $S_i\sub S_{i-1}$
of size $Cn\log n/(3/2)^i$ for sufficiently large constant $C$. Let us use $A_{i, h}$ (resp. $A_{i, h}'$)
for the matrix of dimension $|S_i|\times |V|$ (resp. $|V|\times
|S_i|$) with entries given by $A_{i, h}[s, v] = d_{\leq h}(s, v)$
(resp. $A_{i, h}'[v, s] = d_{\leq h}(v, s)$), where $s\in S_i$ and
$v\in V$. Note in particular that $A_{0, 1} = A_{0, 1}' = W$.

We proceed in $\log_{3/2} n$ iterations, starting with iteration
$1$. In iteration $i$, we use the matrices $A_{i-1, h'}$ and $A_{i-1,
  h'}'$ for $h'\in [(3/2)^{i-1}]$ to compute the matrices $A_{i, h}$
and $A_{i, h}'$ for all $h\in[(3/2)^i]$. Let $h\in [(3/2)^i]$. We only describe how to compute $A_{i, h}$ as the computation of
$A_{i, h}'$ is symmetric.

If $h\leq (3/2)^{i-1}$, then, since $S_i\sub S_{i-1}$, we can extract
$A_{i, h}$ from $A_{i-1, h}$. Now assume $(3/2)^{i-1} < h\leq
(3/2)^i$. Define the matrix $B$ of dimension $(1/2)(3/2)^{i-1}|S_i|\times
(3/2)^{i-1}|S_{i-1}|$ as follows: if $h-h'\leq (3/2)^{i-1}$, then $B[(s, h), (s', h')] = A_{i-1, h-h'}[s, s']$, and otherwise, $B[(s, h), (s', h')] = \infty$, for $s\in S_i$, $s'\in S_{i-1}$, $(3/2)^{i-1} < h\leq
(3/2)^i$, and $h'\in [(3/2)^{i-1}]$. Define the matrix $C$ of dimension $(3/2)^{i-1}|S_{i-1}|\times |V|$ by $C[(s', h'), v] = A_{i-1, h'}[s',
  v]$ for $s'\in S_{i-1}$, $h'\in [(3/2)^{i-1}]$, and $v\in V$. Let $D = B\star C$, and we describe how to obtain $A_{i, h}$
using $D$.

For any $s\in S_i$ and $v\in V$, let $p$ be a simple path achieving $d_{\le h}(s, v)$. If $|p|\leq (3/2)^{i-1}$, then we
have already computed $d_{\leq h}(s, v)$ prior to the start of the
iteration. Otherwise, let us define a middle third of this path $p$ to
be a subpath $p_2 = (u_1, \ldots, u_2)$ of length $
(1/3)(3/2)^i$ such that we can write $p$ as the concatenation of $p_1, p_2, p_3$, where $p_1 = (s, \ldots, u_1)$ and $p_3 = (u_2, \ldots, v)$ are
subpaths of $p$ of lengths at most $(1/3)(3/2)^i$. For any middle third $p_2$ of $p$, by Lemma \ref{lemma:hitting_set}, w.h.p.,
there is a vertex $s'$ in $S_{i-1}$ that is on $p_2$; and furthermore,
by definition of middle third, we have that both of the subpaths $(s,
\ldots, s')$ and $(s', \ldots, v)$ of $p$ have hop-lengths at most
$(3/2)^{i-1}$. This implies the following:
\begin{align*}
d_{\leq h}(s, v) &= \min\left\{d_{\leq (3/2)^{i-1}} (s, v), \min_{\substack{h'\in [(3/2)^{i-1}]: h - h' \le (3/2)^{i-1} \\ s'\in S_{i-1}}} \left\{d_{\leq h-h'}(s, s') + d_{\leq h'}(s', v)\right\}\right\}\\
&=\min\left\{d_{\leq (3/2)^{i-1}} (s, v), \min_{h'\in [(3/2)^{i-1}], s'\in S_{i-1}} \left\{B[(s, h), (s', h')]+C[(s', h'), v]\right\}\right\}\\
&= \min\left\{A_{i-1, (3/2)^{i-1}}[s, v], D[(s, h), v]\right\}.
\end{align*}
Thus, we have $A_{i, h}[s, v] = \min\{A_{i-1, (3/2)^{i-1}}[s, v], D[(s, h), v]\}$.

Now we discuss how to answer a query $(u, v, h)$, where $u, v\in V$ and $h\in [n]$.
Let $j^*$ be the unique integer satisfying $(3/2)^{j^*-1}< h \le
(3/2)^{j^*}$, then
\[d_{\leq h}(u, v) = \min_{j\leq j^*, h'\in [\min\{h, (3/2)^j\}], s'\in S_j}\left\{A_{j, \min\{h-h', (3/2)^j\}}[u, s']+A_{j, h'}'[s', v]\right\}.
\]

Suppose that $d_{\leq h}(u,v)$ is achieved by a simple $b$-hop path $P$ for some $b\leq h$ and let $j$ be the unique integer satisfying $(3/2)^{j-1} < b \le (3/2)^{j}$.
By Lemma \ref{lemma:hitting_set}, with high probability,
there is a vertex in $S_j$ that lies on $P$, and correctness follows.

In each iteration, we do a min-plus product between matrices of
dimension $\widetilde{O}(n)\times\widetilde{O}(n)$, which is
$n^{3}/2^{\Theta(\sqrt{\log n})}$ time; summing over all $\log_{3/2} n$
iterations, we obtain that the preprocessing time is
$n^{3}/2^{\Theta(\sqrt{\log n})}$ and that the space used is
$\widetilde{O}(n^2)$.

The query time is, within constant factors,
\[\sum_{j\leq j^*} (3/2)^{j+1}n\log n/(3/2)^j\leq O(n\log^2 n)= \widetilde{O}(n).\]
\end{proof}

Finally, we show a more efficient distance oracle for graphs with integer edge weights in $\{-M, \ldots, M\}$. 

\DObounded*

\begin{proof}
    For every $k \ge 0$ where $K=(3/2)^k \le n$, we will sample a random subset $S_k \subseteq V$ of size $\frac{C n \log n}{K}$ for sufficiently large constant $C$ and we furthermore require $S_k \subseteq S_{k-1}$ for every $k$, and compute and store distances $d_{\le h}(u, s)$ and $d_{\le h}(u, s)$ for every $u \in V, s \in S_k, h \in [K]$. Clearly, the space usage is $\tO(n^2)$, and the $\tO(n)$ query time follows from essentially the same query algorithm from \cref{thm:do-mn}. Therefore, in the following, we focus on the preprocessing time, and we provide two alternative algorithms. Also, we focus on computing $d_{\le h}(s, u)$, as $d_{\le h}(u, s)$ is symmetric. 

    The first algorithm computes $d_{\le h}(s, u)$ for every $u \in V, s \in S_{k+1}, h \in [3K/2]$ assuming that  $d_{\le h}(s, u)$ for every $u \in V, s \in S_{k}, h \in [K]$ are given. For any $u \in V, s \in S_{k+1}, h \in [3K/2]$, let $p$ be a simple path achieving $d_{\le h}(s, u)$. If $|p| \le K$, then $d_{h}(s, u) = d_K(s, u)$, which is given as $S_{k+1} \subseteq S_k$. Otherwise, with high probability, there exists $x \in S_k$ that splits $p$ into two subpaths $p_1$ (from $s$ to $x$) and $p_2$ (from $x$ to $u$), so that $|p_1|, |p_2| \le K$. Therefore, similar to the proof of \cref{thm:ap-upper-bound}, in order to compute $\left\langle d_h(s, u)\right\rangle_{h=K+1}^{3K/2}$, it suffices to compute the min-plus convolution between $\left\langle d_h(s, x)\right\rangle_{h=1}^{K}$ and $\left\langle d_h(x, u)\right\rangle_{h=1}^{K}$ for every $x \in S_k$. Note that the input sequences are non-increasing, and are bounded by $O(KM)$, so each min-plus convolution can be computed in $\tO(K \sqrt{KM})$ time by \cite[Theorem 7]{BDP24}, and the total number of min-plus convolutions we need to compute is $O(|S_{k+1}| |V| |S_k|) = \tO(n^3 / K^2)$. Therefore, the overall time complexity is $\tO(M^{0.5} K^{-0.5} n^3)$. 

    Next, we provide an alternative algorithm for computing $d_{\le h}(s, u)$ for every $u \in V, s \in S_{k}, h \in [K]$ that is more efficient for smaller values of $K$. In fact, we have already shown a similar algorithm in the proof for \cref{thm: single-source}, but we repeat here for completeness. Let $A_h$ be a matrix indexed by $S_k \times V$ where $A_h[s, v] = d_h(s, u)$. Clearly, if we compute $A_1, \ldots, A_K$, we can compute the desired distances easily. Let $W$ be the weighted adjacency matrix of the graph, then it is not difficult to see that $A_{h+1} = A_h \star W$. First, we use repeated min-plus product squaring to compute (we round $K$ to the next power of $2$ WLOG)
    \[
    W, W^2, W^4, \ldots, W^{K}. 
    \]
    Each min-plus product is between $n \times n$ matrices with entries bounded by $O(KM)$, so the running time is $\tO(KMn^\omega)$ by \cref{lem:small-mpp}. 

    Next, for any $0 \le i \le \log(K) - 1$, if we stack $A_1, A_2, \ldots, A_{2^i}$ together, and multiply the stacked matrix with $W^{2^i}$, we will obtain a matrix that is $A_{2^i+1}, A_{2^i+2}, \ldots, A_{2^{i+1}}$ stacked together. The stacked matrix has dimension $O(|S_k| \cdot 2^i) = O(|S_k| \cdot K) = \tO(n)$, and the entries are integers bounded by $O(KM)$, so each min-plus product can be computed in $\tO(KM n^\omega)$ time by \cref{lem:small-mpp}. Thus, the overall, running time is $\tO(KM n^\omega)$. 

    We can choose one of the two algorithms depending on whether $K$ is big or small. The worst-case running time is attained when $K =n^{2-2\omega / 3} / M^{1/3}$, and the final preprocessing time is $\tO(M^{2/3} n^{(6+\omega)/3})$. 
\end{proof}

%% file: lower-bound-weighted.tex
\section{Conditional Lower Bounds for Weighted Graphs}

In this section, we prove \cref{thm:weighted-lower-bound}, which we first recall below:

\APLowerBound*

We first define the following intermediate problem: 
\begin{problem}
\label{prob:prob1}
Given two $n \times n$ matrices $A, B$, compute 
$$\min_{x + y = \ell} \left\{ A[i, x] + B[j, y]\right\}$$
for every $i, j \in [n], \ell \in [2n]$. 
\end{problem}

\begin{lemma}
Under the Min-Plus Convolution hypothesis, \cref{prob:prob1} requires $n^{4-o(1)}$ time. 
\end{lemma}
\begin{proof}
Given a Min-Plus Convolution instance between two length-$N$ sequences $X, Y$, we split these two sequences into $\sqrt{N}$ subsequences $X_1, \ldots, X_{\sqrt{N}}, Y_1 \ldots, Y_{\sqrt{N}}$. Then it suffices to compute the Min-Plus Convolution between $X_i$ and $Y_j$ for every $i, j$. Letting $A[i, x] = X_i[x]$, and $B[j, y] = Y_j[y]$, this is exactly an instance of \cref{prob:prob1} of size $n = \sqrt{N}$. Thus, if one can solve \cref{prob:prob1} in $O(n^{4-\eps})$ time for $\eps > 0$, then one can also solve the Min-Plus Convolution instance in $O(N^{2-\eps/2})$ time, contradicting the Min-Plus Convolution hypothesis. 
\end{proof}

\begin{lemma}
\label{lem:weighted-exact-lower-bound}
Under the Min-Plus Convolution hypothesis, computing $d_{\ell}(i, j)$ in a weighted graph $G = (V, E, w)$ for every $i, j \in V, \ell \in [|V|]$ requires $|V|^{4-o(1)}$ time. 
\end{lemma}
\begin{proof}
We reduce an instance of \cref{prob:prob1} of dimension $n$ to an instance of the problem of computing $d_{\ell}(i, j)$ on a graph $G=(V,E, w)$ with $O(n)$ vertices for every pair of vertices $i, j$ and every $\ell \in [|V|]$. This way, as \cref{prob:prob1} requires $n^{4-o(1)}$ time under the Min-Plus Convolution hypothesis, the  problem of computing $d_{\ell}(i, j)$ requires $|V|^{4-o(1)}$ time under the Min-Plus Convolution hypothesis. 

The graph will consist of five layers, where all layers except the third layer have $n$ nodes, and the third layer contains a singleton $s$. The first, second, fourth, fifth layers represent $i, x, y, j$ respectively. 

Between the first and second layer, we add edge $(i, x)$ with weight $A[i, x]$. Inside the second layer, we add an edge $(x, x-1)$ for every $2 \le x \le n$ with weight $0$. Between the second and third layer, we add an edge $(1, s)$ with weight $0$. This way, the only path from $i$ to $s$ with $x+1$ edges has weight $A[i, x]$. We similarly create the right half of the graph: between the fourth and fifth layer, we add an edge $(y, j)$ with weight $B[j, y]$; inside the fourth layer, we add an edge $(y - 1, y)$ for every $2 \le y \le n$ with weight $0$; between the third and fourth layer, we add an edge $(s, 1)$ with weight $0$. The only path from $s$ to $j$ with $y+1$ edges has weight $B[j, y]$. 

Thus, the shortest path from $i$ to $j$ with $\ell + 2$ edges will have weight $\min_{x + y = \ell} \left\{ A[i, x] + B[j, y]\right\}$. 

This also implies that if we preprocess the graph in $O(n^{4-\eps})$ time for some $\eps>0$, then the $n^3$ queries given by $(i,j,\ell)$ must take $n^{1-o(1)}$ time (amortized).
\end{proof}

Given \cref{lem:weighted-exact-lower-bound}, the proof of \cref{thm:weighted-lower-bound} follows from the equivalence between the ``exact $h$'' and ``at most $h$'' versions of the problem for weighted graphs, as mentioned in \cref{sec:intro}. We nevertheless still repeat the argument for completeness. Also, note that the graph we constructed is a DAG. 

\begin{proof}[Proof of \cref{thm:weighted-lower-bound}]
    Given any $n$-node graph $G = (V, E, w)$ for which we need to compute $d_{\ell}(i, j)$ for $i, j \in V$, $\ell \in [n]$, we can construct a new graph $G'$ that is a copy of $G$ but we add $-M$ to the weight of every edge, for some sufficiently large value $M$ (say $M$ is greater than $2n$ times the maximum absolute value of all edge weights). Then suppose we compute $d'_{\le \ell}(i, j)$, which is the distance from $i$ to $j$ in $G'$ using at most $\ell$ hops. As $M$ is large enough, the shortest such path must use exactly $\ell$ edges (assuming such a path exists), so we can add $\ell M$ to $d'_{\le \ell}(i, j)$ to recover $d_{\ell}(i, j)$. We will also be able to detect the case where there is no path with exactly $\ell$ hops from $i$ to $j$, as in that case, we must have $d'_{\le \ell}(i, j) + \ell M \ge M / 2$, which is impossible if such a path exists. 
\end{proof}

%% file: BMMlb.tex
Here we show that Single Pair All-Hops Shortest Paths requires $n^{\omega-o(1)}$ time, assuming that Triangle Detection does. The proof is similar in spirit to proofs in prior work (See \cite{focs10,KoAd23,BringmannGKL24}).

\begin{theorem}
If one can solve the Single Pair All-Hops Shortest paths Problem in graphs with $n$ vertices and edge weights in $\{-1,1\}$ in $O(n^{\omega-\eps})$ time for some $\eps>0$, then one can also find a triangle in an $n$ node graph in $O(n^{\omega-\eps})$ time.
\end{theorem}

\begin{proof}
Let $H=(V_H,E_H)$ be a given graph in which we want to find a triangle.
Without loss of generality, $H$ is tripartite on parts $I,J,K$ on $n$ nodes each and we want to determine whether there are $i\in I,j\in J, k\in K$ such that $(i,j),(j,k),(k,i)\in E_H$.

We will create a graph $G=(V,E,w)$ whose weights are in $\{-1,1\}$ and $s,t\in V$ such that given $d_{\leq h}(s,t)$ for all $h$ we can determine whether $H$ contains a triangle.

First, create two copies of $I$: $I_1$ and $I_2$. Then create a copy $J'$ of $J$ and $K'$ of $K$. Finally, create two nodes $s$ and $t$.

The vertices of $G$ are $V=I_1\cup J'\cup K'\cup I_2\cup \{s,t\}$.

Let the nodes of $I$ in $H$ be $i_1,\ldots, i_n$. Call their copies in $I_j$ (for $j=1,2$) $i^j_1,\ldots,i^j_n$.

We will now describe the directed edges of $G$.

First, there is a path $P_1:=s\rightarrow i^1_1\rightarrow i^1_2\rightarrow \ldots \rightarrow i^1_n$ and all the weights of the edges in the path are $-1$.

Similarly, there is a path $P_2:=i^2_1\rightarrow i^2_2\rightarrow \ldots \rightarrow i^2_n\rightarrow t$ and all the weights of the edges in the path are $-1$.

For every edge $(i_p,j)$ of $H$ with $i_p\in I,j\in J$, we have the edge $(i_p^1,j')$ where $j'$ is the copy of $j$ in $J'$. 
For every edge $(j,k)$ of $H$ with $j\in J, k\in K$, we have the edge $(j',k')$ where $j'$ and $k'$ are the copies of $j$ in $J'$ and $k$ in $K'$ respectively.
For every edge $(k,i_p)$ of $H$ with $i_p\in I,k\in K$, we have the edge $(k',i_p^2)$ where $k'$ is the copy of $k$ in $K'$. The weights of all these edges are $1$.

Now let us consider a path from $s$ to $t$ in $G$.
As $G$ is a DAG, the only types of paths look like this:
They first follow the path $P_1$ from $s$ to some node $i^1_p$, then they go across using three edges from $i^1_p$ through $J'$ and $K'$ and then to some $i^2_q$ and then down the path $P_2$ all the way to $t$.

For such a path, the number of hops is 
$$p+3+(n+1-q)=(n+4)+(p-q).$$
The weight of the path is
$$-p+3-(n+1-q)=(2-n)-(p-q).$$

If we focus on $d_{\leq n+4}(s,t)$, then the paths of consideration must have $p\leq q$.
Meanwhile, to minimize the weight over all such paths, one must maximize $(p-q)$ subject to $p\leq q$. Thus the minimum possible weight $2-n$ is achieved for $p=q$, and this is possible if and only if there exists some $p$ such that there are $j\in J,k\in K$ and edges $(i_p,j),(j,k),(k,i_p)$ in $H$. That is, $d_{\leq n+4}(s,t)=2-n$ if and only if there is a triangle in $H$ (and it is larger otherwise).
\end{proof}